\documentclass{article}

     \PassOptionsToPackage{numbers, compress}{natbib}

\usepackage[preprint]{neurips_2023}

\usepackage[utf8]{inputenc} 
\usepackage[T1]{fontenc}    
\usepackage{hyperref}       
\usepackage{url}            
\usepackage{booktabs}       
\usepackage{amsfonts}       
\usepackage{nicefrac}       
\usepackage{microtype}      
\usepackage{xcolor}         

\usepackage{bbold}
\usepackage{amsmath}
\usepackage{amsfonts}
\usepackage{mathtools}
\usepackage{tikz}
\usetikzlibrary{calc,trees,positioning,arrows,fit,shapes,calc}
\usepackage[english]{babel}
\usepackage{algorithm2e}
\usepackage{tabularx}
\usepackage{amsmath}
\usepackage{amsthm}

\newcommand\Set[2]{\Bigl\{\,#1\mid#2\,\Bigl\}}

\usepackage{amsfonts} 
\usepackage{stmaryrd}
\usepackage{xparse} \DeclarePairedDelimiterX{\Iintv}[1]{\llbracket}{\rrbracket}{\iintvargs{#1}}

\NewDocumentCommand{\iintvargs}{>{\SplitArgument{1}{,}}m}
{\iintvargsaux#1} 

\NewDocumentCommand{\iintvargsaux}{mm} {#1\mkern1.5mu,\mkern1.5mu#2}

\theoremstyle{theorem}
\newtheorem{theorem}{Theorem}[section]
\newtheorem{corollary}{Corollary}[theorem]
\newtheorem{lemma}[theorem]{Lemma}
\theoremstyle{definition}
\newtheorem{definition}{Definition}[section]
\theoremstyle{interpretation}
\newtheorem{interpretation}{Interpretation}[section]
\newtheorem*{remark}{Remark}

\title{Computing the Vapnik Chervonenkis Dimension for Non-Discrete Settings }
\author{%
  Mohamed MOUHAJIR\thanks{These authors contributed equally to this research.} \\
  ENSIAS, Mohammed V University \\
  Morocco, Rabat 10000 \\
  \texttt{mohamed\_mouhajir@um5.ac.ma}
  \And
  Mohammed NECHBA\footnotemark[1] \\
  ENSIAS, Mohammed V University \\
  Morocco, Rabat 10000 \\
  \texttt{mohammed\_nechba@um5.ac.ma}
  \And
  Yassine SEDJARI\footnotemark[2] \\
  ENSIAS, Mohammed V University \\
  Morocco, Rabat 10000 \\
  \texttt{yassine\_sedjari@um5.ac.ma}
}

\begin{document}
\maketitle

\begin{abstract}

In 1984, Valiant \cite{valiant1984theory} introduced the Probably Approximately Correct (PAC) learning framework for boolean function classes.  Blumer et al. \cite{Blumer89} extended this model in 1989 by introducing the VC dimension as a tool to characterize the learnability of PAC. The VC dimension was based on the work of Vapnik and Chervonenkis in 1971 \cite{Vapnik68}, who introduced a tool called the growth function to characterize the shattering property. Researchers have since determined the VC dimension for specific classes, and efforts have been made to develop an algorithm that can calculate the VC dimension for any concept class. In 1991, Linial, Mansour, and Rivest \cite{linial1991results} presented an algorithm for computing the VC dimension in the discrete setting, assuming that both the concept class and domain set were finite. However, no attempts had been made to design an algorithm that could compute the VC dimension in the general setting.Therefore, our work focuses on developing a method to approximately compute the VC dimension without constraints on the concept classes or their domain set. Our approach is based on our finding that the Empirical Risk Minimization (ERM) learning paradigm can be used as a new tool to characterize the shattering property of a concept class.

\end{abstract}

\section{Introduction}

One critical facet of Statistical Learning Theory is to prove the learnability of some concept classes within the PAC (Probably Approximately Correct) learning framework, which was introduced by Valiant in 1984 \cite{valiant1984theory}. To this end, Blumer et al. extended in 1989 \cite{Blumer89} Valiant's work by coining the notion of Vapnik-Chervonenkis (VC) dimension as a tool that characterizes the PAC learnability of binary-valued concept classes. This notion of VC dimension was buid upon  the work of Vapnik and Chervonenkis in 1971 \cite{Vapnik68} , who introduced a tool called the growth function to characterize the shattering property. The main result showed by Blumer et al. \cite{Blumer89} was that finiteness of VC dimension is a necessary and sufficient condition to say that a concept class is PAC learnable.

     As time passed, researchers directed their interest towards calculating the VC dimension for some specific cases. In this context, the work of Dudley et al. in 1981 \cite{WENOCU80} has laid the foundation ground of what was called lately by Ben-David \& Shalev-Shwartz ,2014 \cite{David14} “Dudley Classes”. Due to this notion, the VC dimension of many classes was easily calculated. Examples of such classes include: the class of Half-Spaces, the class of Open Balls, and the class of polynomials that are of degree less than some given integer, and so on. Subsequently, attention has shifted towards designing an algorithm that could compute the VC dimension of any concept class $\mathcal{H}$. In this scope, the only formal initiative to develop such an algorithm, to our knowledge, was taken by Linial et al. in 1991 \cite{linial1991results}, when they constructed an algorithm to calculate the VC dimension in what they referred to as the 'Discrete VC problem'. Yet, two key assumptions underpinned their approach: that the concept class and domain set are both finite. Meaning that we can present the concept class as $ \mathcal{H} = \{ h_1, ..., h_l \}$ , and the domain set as $\mathcal{X}= \{ x_1, ..., x_m \}$.
    
  Then, according to this modeling, the concept class and the domain set are presented as a $0-1$ valued matrix $M$. Each row $i$, represents a concept $h_i$ in $\mathcal{H}$, and each column $j$ represent a domain point $x_j$. Then, each element $M_{i,j}$ in this matrix is filled as following:
  \begin{itemize}
      \item If the $x_j$ belongs to $h_i$ , then $M_{i,j}= 1$.
\item Otherwise , $M_{i,j}= 0$.

  \end{itemize}

Therefore, in this context, computing the VC dimension of a concept class $\mathcal{H}$ becomes the problem of computing the VC dimension of the above matrix $M$ according to the following technique: To prove that the VC dimension is less than d, we should take all combinations of columns $\binom{m}{d}$. Each combination represents a set upon which we will assess the shattering property. Thus, if we find that $\mathcal{H}$ doesn’t shatter any of them , we say that $VC dim (\mathcal{H}) < d$ .

However, it is evident that this approach lacks generality. In fact, the most practical concept classes are infinite. Besides, if we assume the finiteness of the concept class, then computing the VC dimension of this concept class is not necessary to prove that it is PAC learnable. As stated by Ben-David and Shalev-Shwartz (2014)\cite{David14} , the sample complexity required to ensure PAC learnability is directly proportional to the size of the concept class: $m(\epsilon,\delta) \leq \frac{\ln{(|\mathcal{H}|} / \delta)}{\epsilon}$.

 In this work, our aim is to address this gap by proposing a new algorithm that can compute the VC dimension of a concept class without requiring the finiteness of either the class or its domain set. This approach has the potential to significantly expand the scope of applicability of the VC dimension concept, making it relevant for a wider range of real-world scenarios.

\begin{remark}
In this article, we will mainly use the common terminology and notation as in the textbook \cite{David14}. 
We will use the term shattering to refers to Vapnik-Chervonenkis shattering (since there are more sophisticated approaches of defining the shattering property which are used for mutliclass problems such us : Natarajan shattering and DS shattering \cite{Brukhim22})
\end{remark}

\subsection{Results}

Our main result is the characterization of the shattering property using the ERM (Empirical Risk Minimization) learning paradigm. Specifically, we have derived a theorem and a corollary that provide an algorithmic way to assess the shattering property of any hypothesis class.This result has significant practical benefits, allowing us to efficiently evaluate the shattering property of any given class:

\begin{theorem}[Nechba-Mouhajir theorem of shattering]
Let $\mathcal{H}$ be a class of functions from $\mathcal{X}$ to $\{0, 1\}$ . Let $\mathcal{C} =\{ c_1,...,c_d\} \subset \mathcal{X}$. Let $(y_1,...,y_d) \in \{0,1\}^d$ a fixed $d$-tuple ,for some integer $d$.\\
There exists a function $h^* \in \mathcal{H}$ such that $h^*(\mathcal{C}) = (y_1,...,y_d)$ if and only if the $ERM_\mathcal{H}$ over $\mathcal{S}=\{(c_1,y_1),...,(c_d,y_d)\}$ (where $ERM_\mathcal{H}(\mathcal{S}) = \underset{{h\in \mathcal{H}}}{\arg\min(L_{\mathcal{S}} (h))}$ ) gives zero empirical loss  ( i.e : $L_{\mathcal{S}}(ERM_{\mathcal{H}}(S))=0$ ).
\end{theorem}

\begin{corollary}
Let $\mathcal{H}$ be a hypothesis class for binary classification. Let $\mathcal{C}=\{c_1,...,c_d\}\subset \mathcal{X}$.\\
$\mathcal{H}$ shatters $\mathcal{C}$ if and only  if for every $d$-tuple in $\{0,1\}^{d}$ (i.e $\forall (y_1,...,y_d)\in \{0,1\}^d $) , the $ERM_\mathcal{H}$ learner over $\mathcal{S}=\{(c_1,y_1),...,(c_d,y_d)\}$ will give a zero empirical loss.
\end{corollary}

\section{ERM as a new characterization of the shattering property}

\subsection{The Problem of Shattering : Definitions and Properties }
Let $\mathcal{H}$ be a hypothesis class for  binary classification task :

\begin{align*}
\mathcal{H} &\coloneqq \Set{\emph{h}}{\operatorname{\emph{h}}\colon \mathcal{X}\to  \mathcal{Y}= \{0,1\},\quad x \mapsto \emph{h}(x)} 
\end{align*}

\begin{definition}[The restriction of $\mathcal{H}$ over a domain $\mathcal{C}$]
Let $\mathcal{H}$ be a class of functions from $\mathcal{X}$
to $\{0,1\}$ and let $\mathcal{C} =\{ c_1,...,c_d\} \subset \mathcal{X}$. The restriction of $\mathcal{H}$ to $\mathcal{C}$ is the set of functions from $\mathcal{C}$ to $\{0,1\}$ that can be derived from $\mathcal{H}$. That is,
\begin{align*}
    \mathcal{H}_c=\Set{ (\emph{h}(c_1),..., \emph{h}(c_d))}{ \emph{h} \in \mathcal{H}}
\end{align*}
where we represent each function from $\mathcal{C}$ to $\{0,1\}$ as a vector in $\{0,1\}^{|\mathcal{C}|}$
\end{definition}

\begin{definition}[Shattering]
A hypothesis class $\mathcal{H}$ of functions from $\mathcal{X}$ to $\{0, 1\}$ shatters a finite set $\mathcal{C} \subset X$ if the restriction of $\mathcal{H}$ to $\mathcal{C}$ is the set of all functions from $\mathcal{C}$ to $\{0,1\}$. That is, $ |\mathcal{H}_c|= 2^{|\mathcal{C}|}$.
\end{definition}

\begin{interpretation}
To prove that a given hypothesis class $\mathcal{H}$ shatters a given domain set $\mathcal{C}\subset \mathcal{X}$, we should perform two steps:
\begin{enumerate}
 
    \item $\Rightarrow |$ 
  We should prove that $\mathcal{H}_c \subset \{0,1\}^{|\mathcal{C}|}$ (which is trivial since $\mathcal{Y}=\{0,1\})$
  
    \item $\Leftarrow |$ (see Figure\ref{fig:explicative_diagram} for more intuition) We should prove that $\{0,1\}^{|\mathcal{C}|} \subset \mathcal{H}_c$. In other words, we should prove that for every $|\mathcal{C}|$-tuple of 0 and 1, there exist a hypothesis $h^*$ such that applying $h^*$ over $\mathcal{C}$ will give as this $|\mathcal{C}|$-tuple. Formally :
    \begin{align*}
          \{0,1\}^{|\mathcal{C}|} \subset \mathcal{H}_c \Longleftrightarrow \forall \;(y_1,...,y_{|\mathcal{C}|}) \in \{0,1\}^{|\mathcal{C}|},\; \exists \emph{h}^* \in  \mathcal{H}: \emph{h}^*(\mathcal{C}) &= \emph{h}^*((c_1,...,c_{|\mathcal{C}|}))\\
        &= (\emph{h}^*(c_1),...,\emph{h}^*(c_{|\mathcal{C}|}))\\
        &= (y_1,...,y_{|\mathcal{C}|})
    \end{align*}
\end{enumerate}

\begin{figure}[htp]
\centering
\begin{tikzpicture}[ele/.style={fill=black,circle,minimum width=.8pt,inner sep=0.5pt},every fit/.style={ellipse,draw,inner sep=-2pt}]

    \node at (0,5) {$\mathcal{H}_\mathcal{C}$};
    \node at (6,5) {$\{0,1\}^{|\mathcal{C}|}$};

 
  \node[ele,label=$h_1^*(\mathcal{C})$] (a1) at (0,4) {};    
  \node[ele,label] (a2) at (0,3.5) {};    
  \node[ele,label] (a3) at (0,3.2) {};   
  \node[ele,label] (a4) at (0,2.9) {};   
  \node[ele,label=$h_n^*(\mathcal{C})$] (a5) at (0,2) {};
  \node[ele,label] (a6) at (0,1.5) {};
   \node[ele,label] (a7) at (0,1.2) {};
    \node[ele,label] (a8) at (0,0) {};

  \node[ele,,label=\text{$(y_{1},...,y_{|\mathcal{C}|})$}] (b1) at (6,4) {};
  \node[ele,,label] (b2) at (6,3) {};
  \node[ele,,label] (b3) at (6,2.5) {};
  \node[ele,,label] (b4) at (6,2) {};
  \node[ele,,label=right:\text{$(y_{1}^{,},...,y_{|\mathcal{C}|}^{,})$}] (b5) at (6,0) {};

  \node[draw,fit= (a1) (a2) (a7) (a8),minimum width=3.7cm] {} ;
  \node[draw,fit= (b1) (b2) (b4) (b5),minimum width=3.7cm] {} ;  
  \draw[->,thick,shorten <=2pt,shorten >=2pt] (b5) -- (a1);
  \draw[->,thick,shorten <=2pt,shorten >=2] (b1) -- (a5);

\end{tikzpicture}
\caption{Explicating diagram of $\{0,1\}^{|C|} \subset \mathcal{H}_{\mathcal{C}} $  }

\label{fig:explicative_diagram}
\end{figure}
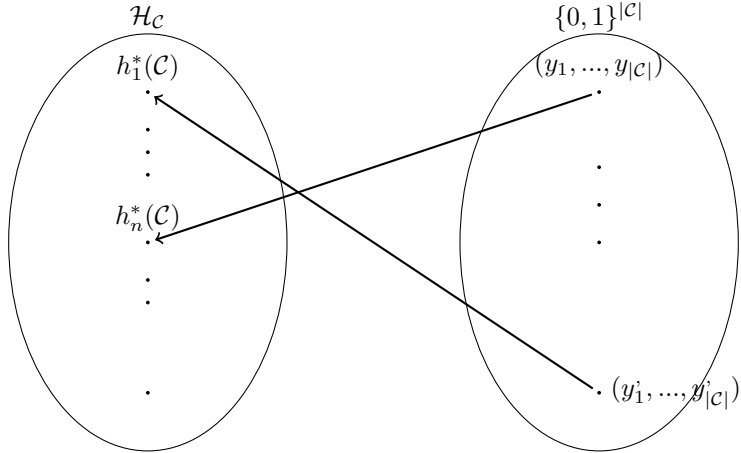

\end{interpretation}

Now , we will present our first contribution in this paper which this the following theorem of shattering :

\begin{theorem}[Nechba-Mouhajir theorem of shattering]
Let $\mathcal{H}$ be a class of functions from $\mathcal{X}$ to $\{0, 1\}$ . Let $\mathcal{C} =\{ c_1,...,c_d\} \subset \mathcal{X}$. Let $(y_1,...,y_d) \in \{0,1\}^d$ a fixed $d$-tuple, for some integer $d$.\\
There exists a function $h^* \in \mathcal{H}$ such that $h^*(\mathcal{C}) = (y_1,...,y_d)$ if and only if the $ERM_\mathcal{H}$ over $\mathcal{S}=\{(c_1,y_1),...,(c_d,y_d)\}$ (where $ERM_\mathcal{H}(\mathcal{S}) = \underset{{h\in \mathcal{H}}}{\arg\min(L_{\mathcal{S}} (h))}$ ) gives zero empirical loss  ( i.e : $L_{\mathcal{S}}(ERM_{\mathcal{H}}(S))=0$ ) .
\end{theorem}
\begin{proof}
\begin{itemize}
    \item $\Rightarrow |$ In fact, if there exists an $h^* \in \mathcal{H}$ such that $h^*(\mathcal{C})= h^*(\{c_1,...,c_d\})=(y_1,...,y_d)$ then $ ERM_\mathcal{H}(\mathcal{S})=h^*$ and , therefore ,  $  L_\mathcal{S}(ERM_\mathcal{H}(\mathcal{S}))= L_\mathcal{S}(h^*)=0 $.

\item $\Leftarrow |$ Let $h_0 = ERM_\mathcal{H}(\mathcal{S}) = \underset{{h\in \mathcal{H}}}{\arg\min(L_{\mathcal{S}} (h))}$ s.t: $\mathcal{S}=\{(c_1,y_1),...,(c_d,y_d)\}$.\\
Suppose that $L_\mathcal{S}(h_0) = 0$\\
Given that $ \begin{cases}
    L_\mathcal{S}(h_0) = \frac{1}{d} \sum_{i=1}^d l(h_0(c_i),y_i)\\
    \forall i \in \llbracket 1, d \rrbracket \; \quad l(h_0(c_i),y_i)  \geq 0 
\end{cases}
$\\

\text{ ( By definition, a loss function -in prediction problems \\ ( i.e. $\mathcal{Z}=\mathcal{X}\times \mathcal{Y }$)- is : $ l \text{ : } \mathcal{H} \times \mathcal{Z} \rightarrow \mathbb{R}_{+}$) }

Then $ \forall i \in \llbracket 1, d \rrbracket$ : $l(h_0(c_i),y_i)=0$.\\

Therefore : $ \forall i \in \llbracket 1, d \rrbracket$ : $h_0(c_i)=y_i$. Namely, there exists $h^*= h_0= ERM_\mathcal{H}(\mathcal{S})$ such that: $ \forall i \in \llbracket 1, d \rrbracket$ :$h^*(c_i)=y_i$

\end{itemize}
\end{proof}

\begin{corollary}
Let $\mathcal{H}$ be a hypothesis class for binary classification. Let $\mathcal{C}=\{c_1,...,c_d\}\subset \mathcal{X}$.\\
$\mathcal{H}$ shatters $\mathcal{C}$ if and only  if for every $d$-tuple $(y_1,...,y_d)$ in $\{0,1\}^{d}$,   the $ERM_\mathcal{H}$ over $\mathcal{S}=\{(c_1,y_1),...,(c_d,y_d)\}$ will give a zero empirical loss. Formally: $$\forall (y_1,...,y_d)\in \{0,1\}^d  : L_{\mathcal{S}}(ERM_{\mathcal{H}}(S))=0 , \text{where } \mathcal{S}=\{(c_1,y_1),...,(c_d,y_d)\} $$
\end{corollary}
\begin{proof}
To prove this corollary we will exploit the previous Nechba-Mouhajir theorem of shattering :\\
\begin{align*}
    \mathcal{H} \text{ shatters } \mathcal{C}= 
\{c_1,...,c_d\} & \Leftrightarrow \text{ For every } (y_1,...,y_d)\in \{0,1\}^d \text{, 
 }\exists\; h^*\in \mathcal{H} \text{ s.t}: h^*(\mathcal{C})=(y_1,...,y_d)\\
    & \Leftrightarrow \text{ For every } (y_1,...,y_d)\in \{0,1\}^d, \text{ the } ERM_\mathcal{H} \text{ over } \mathcal{S}=\{(c_1,y_1),...,(c_d,y_d)\} \\ 
 \text{ gives a zero-Empirical Loss.}
\end{align*}

    \end{proof}

\subsection{Nechba-Mouhajir algorithm for shattering}

Now, we will be based on the previous corollary to construct our \textbf{Nechba-Mouhajir algorithm for shattering }:

\SetKwComment{Comment}{/* }{ */}
\RestyleAlgo{ruled} 
\begin{algorithm}
\label{algo_1:Nechba-Mouhajir algorithm for shattering}
\caption{Nechba-Mouhajir algorithm for shattering}
\KwData{
 \begin{itemize}
     \item Hypothesis class $\mathcal{H}$
     \item domain set $\mathcal{C} \subset \mathcal{X}$, (such that $|\mathcal{C}|=d\in \mathbb{N}^*$)
 \end{itemize}
}

\KwResult{
$
\begin{cases}
    \text{True} & \text{if } \mathcal{H} \text{ shatters } \mathcal{C} \\
    \text{False} & \text{if } \mathcal{H} \text{ doesn't shatter } \mathcal{C} 
\end{cases}
$
}

- $All\_combinations \gets $ all d-tuple of $\{0,1\}^{d}$\\

\For{$(y_1,...,y_d)$ in $All\_combinations$}{

\If{$L_\mathcal{S}(ERM_\mathcal{H}(\mathcal{S}))\neq 0$ }{\textbf{Return False}
}

}
\textbf{Return True}

\end{algorithm}

 For a given hypothesis class $\mathcal{H}$ and a given domain set $\mathcal{C} \subset \mathcal{X}$ ( whose size is some $d$), we will generate all possible d-tuples of $\{0,1\}^{d}$. Then, for each d-tuple $(y_1,...,y_d)$, we will consider our sample $\mathcal{S}$ to be the superposition of $\mathcal{C}$ and $(y_1,...,y_d)$. Formally, if we consider $\mathcal{C}=\{c_1,...,c_d\}$ our sample $\mathcal{S}$ will be :
 $$ \mathcal{S} = \{(c_1,y_1),(c_2,y_2),...,(c_d,y_d)\} $$
 
 Then, we will run the ERM rule for $\mathcal{H}$ and over $\mathcal{S}$, and compute the empirical loss of $ERM_\mathcal{H}(\mathcal{S})$ (i.e. $L_\mathcal{S}(ERM_\mathcal{H}(\mathcal{S}))$ ). Then, according to our previous corollary, if we find that $L_\mathcal{S}(ERM_\mathcal{H}(\mathcal{S})) \neq 0 $ we will quit and say that "$\mathcal{H}$ doesn't shatter $\mathcal{C}$". Otherwise, we will continue until we find some adversarial d-tuple (that is, for which $L_\mathcal{S}(ERM_\mathcal{H}(\mathcal{S})) \neq 0 $ ), or we could not find such a d-tuple. In such a case, we can say that our $\mathcal{H}$ shatters $\mathcal{C}$.

\subsection{Overcoming the Practical Issues of the Algorithm}

The proposed algorithm for assessing the shattering property has some practical limitations that need to be addressed in future research. One such limitation is the brute-forcing problem that arises from generating all possible d-tuples of $\{0,1\}^d$. If an adversarial d-tuple is found early in the process, the algorithm can terminate without examining all $2^d$ cases. However, if an adversarial d-tuple is not found, the algorithm will be compelled to check all $2^d$ cases, leading to computational inefficiencies. Alternative methods for generating d-tuples more efficiently can be explored in future work to overcome this brute-force issue.

Another practical challenge is the computational time required by the algorithm, particularly for large values of d. A possible solution to address this issue is to leverage the power of GPUs by assigning a single d-tuple to each of the $2^d$ threads and enabling parallel execution, which can significantly reduce the computation time.  For example, using an NVIDIA GeForce RTX 3080 can allow for nearly $2^{13}$ parallel operations (see NVIDIA's website \cite{nvidia_GeForce_RTX_3080} for more details).  However, for large d, the number of threads required for parallelization may become infeasible. Therefore, future work should investigate alternative approaches to optimize the computational time of the algorithm.

\section{Mouhajir-Nechba Algorithm of VC dimension and its Theoretical Basis}

\subsection{Preliminaries}
Let $ d \in \mathbb{N}^{*}$,  $\mathcal{Y}=\{0,1\}$ and 
\begin{align*}
\mathcal{H} &\coloneqq \Set{h}{h\colon \mathcal{X}\to  \mathcal{Y}= \{0,1\},\quad x \mapsto h(x)} 
\end{align*}

\noindent
Let $Z$ be a random variable, such that :

\begin{equation} \label{eq1}
\begin{split}
Z  & = \text{`` VCdim($\mathcal{H}$) $< d$ ''} \\
             & =  \text{`` $\mathcal{H}$ doesn't shatter any domain set $\mathcal{C}\subset \mathcal{X}$ s.t : $ \mid \mathcal{C}  \mid = d$ "}
\end{split}
\end{equation}

\noindent
The random variable $Z$ has two possible outcomes (i.e. $Z (\Omega) = \{0,1\}$ ) :

$$
Z = \begin{cases}

    1 & \text{ if `` VCdim($\mathcal{H}$) $< d$ ''}= \text{`` For every domain set $\mathcal{C}\subset \mathcal{X}$ s.t : $ \mid \mathcal{C}  \mid = d$ , $\mathcal{H}$ doesn't shatter $\mathcal{C}$ "}\\

    
    0 & \text{ if `` VCdim($\mathcal{H}$) $\geq d$ ''}= \text{`` There exist at least one domain set $\mathcal{C}_{0}\subset \mathcal{X}$ s.t : }
       \text{  $ \mid \mathcal{C}_{0} \mid = d$ and $\mathcal{H}$ shatters $\mathcal{C}_{0}$ "}   
\end{cases}
$$
Let $ \mathbf{p}$ the probability of event ``$Z = 1 $ " , namely $ \mathbf{p = \mathbb{P} ( Z= 1 )}$ . Therefore , $Z$ is a Bernoulli random variable : $Z \sim \mathcal{B}(\mathbf{p})$

\subsection{The estimation of Bernoulli parameter $\mathbf{p}$}

So far, we have proved that $Z \sim \mathcal{B}(\mathbf{p})$. But in order to estimate the parameter $\mathbf{p}$ we will follow the inferential statistics workflow :

\begin{enumerate}
    \item We will take $m$ samples (i.i.d) : $( \mathcal{C}_{1},\mathcal{C}_{2},...,\mathcal{C}_{m} )$ , s.t : $$ \forall i \in \Iintv{1,m} :   \mathcal{C}_{i} \subset \mathcal{X} \quad \text{and} \mid \mathcal{C}_{i}\mid = d $$
    
    \item This is equivalent of taking m (i.i.d) random variables $( Z_{1},Z_{2},..,Z_m)$ that correspond to the realization of $Z$ on the corresponding set $\mathcal{C}_{i}$, for all $i$ in $\Iintv{1,m}$. Formally, we will have $( Z_{1},Z_{2},..,Z_m)$ ,s.t:
       \begin{enumerate}
       \item \text{$Z_{i}$ = `` the realization of $Z$ upon $\mathcal{C}_{i}$ " = `` $\mathcal{H}$ doesn't shatter $\mathcal{C}_{i}$ " }
        \item $Z_{i}=\begin{cases}
            1 &  \text{if } \mathcal{H} \text{ doesn't shatter } \mathcal{C}_{i}    \\
            0 &  \text{if } \mathcal{H} \text{ shatters } \mathcal{C}_{i} 
        \end{cases}$
        \end{enumerate}
        
    \item Thus, according to the law of large numbers: 
    \[ 
    \overline{Z_{m}} = \frac{1}{m} \times \sum_{i=1}^{m} Z_{i} \quad 
    \approx \quad (\mathbf{p}=\mathbb{E}(Z)) \]
\end{enumerate}
\begin{interpretation}
\begin{itemize}
\item $\overline{Z_{m}} $ : is the frequency of the favorable event (i.e.\text{ `` $Z_{i} = 1 $"}  ). Namely  $\overline{Z_{m}} $ is equal to how much time does $\mathcal{H}$ not shatter a set $\mathcal{C}\subset \mathcal{X} $ of size $d$ for our given samples $(\mathcal{C}_{1},\mathcal{C}_{2},...,\mathcal{C}_{m})$.

\item $\overline{Z_{m}} = 1$ : $\mathcal{H}$ doesn't shatter any of sets $\mathcal{C}_{i}$ in our given sample $(\mathcal{C}_{1},\mathcal{C}_{2},...,\mathcal{C}_{m})$.  

\item  $\overline{Z_{m}} \neq 1$ :  there exists at least one set $\mathcal{C}_{0}$ in our sample $(\mathcal{C}_{1},\mathcal{C}_{2},...,\mathcal{C}_{m})$ such that $\mathcal{H}$ shatters $\mathcal{C}_{0}$ . 
\end{itemize}

\end{interpretation}

\subsection{How to choose the number of samples (i.e. $m$) so that $\overline{Z_{m}}$ would be Probably Approximately equal to $\mathbf{p}$ }

Beforehand, we will give a recall of the Hoeffding's inequality \cite{David14} : 

\begin{lemma}[Hoeffding's inequality]
Let $\theta_{1},\theta_{2},...,\theta_{m}$ be a sequence of i.i.d random variables , and assume that for all $i \in \Iintv{1,m}$ , $\mathbb{E}(\theta_{i}) = \mu $ and $\mathbb{P} (a\leq \theta_{i} \leq b) =1$. Then, for any $ \epsilon > 0$:

\[    
  \mathbb{P}( \quad \mid \mu - \frac{1}{m} \times \sum_{i=1}^{m} \theta_{i}    \mid > \epsilon \quad ) \leq 2  \exp{\frac{-2m \epsilon^{2}}{ (b-a)^2 }}
\]

\end{lemma}

As we can see, the Hoeffding's inequality determines, more precisly than the law of large numbers, the convergence rate. By  applying Hoeffding's 
 inequality in our case, we get : 

\begin{equation}\label{Hoeffding's inequality for our case study}
  \mathbb{P}( \quad \mid \mathbf{p} - \frac{1}{m} \times \sum_{i=1}^{m} Z_{i}   \mid  > \epsilon \quad ) \leq 2  e^{-2m \epsilon^{2}}   
\end{equation}

Now, from this equation (\ref{Hoeffding's inequality for our case study}) we can determine ( with respect to the sample size $m$ ) the precision $\epsilon$ and the confidence $(1-\delta = 1-e^{-2m \epsilon^{2}} )$ that we want so as to make the approximation : $ \overline{Z_{m}} \approx \mathbf{p}$ .

In fact, in order to have the precision $\epsilon$ and confidence $(1-\delta)$, we should sample  ($ \mathbf{m= \frac{1}{2\epsilon^{2}} \times \ln{\frac{2}{\delta}} } $) i.i.d sets $\mathcal{C}_{i} $ of size $d$. Then, we can say, with confidence $ (1-\delta)$ and precision $\epsilon$, that \text{ `` $ \overline{Z_{m}} =  \mathbf{p}$ " }.

\begin{interpretation}
 \begin{itemize}
 
    \item If we find that \text{ `` $ \overline{Z_{m}} = 1$ " }. Then, Probably Approximately, \text{ `` $ \mathbf{p} = 1$ "}.\\
    In other words, we can say, with confidence $ (1-\delta)$ and precision $\epsilon$, that : \text{ `` $\mathcal{H}$ doesn't shatter any set $\mathcal{C} \subset \mathcal{X} $ s.t : $ \mid \mathcal{C} \mid=d$  " } which is equivalent of saying \text{ ``VCdim$ (\mathcal{H}) < d$ " }.
    
    \item If we find that \text{ `` $ \overline{Z_{m}} \neq 1$ " }. Then , Probably Approximately, \text{ `` $ \mathbf{p} \neq 1$ "}.\\
    In other words, we can say, with confidence $ (1-\delta)$ and precision $\epsilon$, that : \text{ `` There exists at least one set $\mathcal{C}_{0} \subset \mathcal{X}$ s.t $\mid \mathbf{C} \mid = d $ and $\mathcal{H}$ shatters set $\mathcal{C}_{0}$" }  which is equivalent of saying \text{ ``VCdim$ (\mathcal{H}) \geq d$ " }.
 \end{itemize}
\end{interpretation}

\subsection{Mouhajir-Nechba algorithm for VC-dimension}

Now, after presenting all the theoretical bases, we will unravel our algorithm to compute the Vapnik-Chervonenkis dimension :

\begin{algorithm}
\caption{Mouhajir-Nechba algorithm for VCdimension}\label{alg:two}
\KwData{
 \begin{itemize}
    \item Hypothesis class $\mathcal{H}$
     \item the accuracy $\epsilon$
     \item the confidence parameter $(1-\delta)$
     \item $d_{MAX}$ :  a large number (for example $2^{63}-1$) used to say that $\mathcal{H}$ has infinite VCdim
 \end{itemize}
}

\KwResult{
\[\begin{cases}
    \text{some } d \neq d_{MAX} & \text{if } \mathcal{H} \text{ has finite VCdimension}  \\
d = d_{MAX} & \text{if } \mathcal{H} \text{ has infinite VCdimension}
\end{cases}
\]
}

- Compute the sample size $\mathbf{m}$ ( $ \mathbf{m \gets \frac{1}{2\epsilon^{2}} \times \ln{\frac{2}{\delta}} } $ );\\

\For{$d=1,2,...d_{MAX}$}{
   $\text{Generate an } \mathbf{m} \text{ random samples of size } \mathbf{d} \text{ from } \mathcal{X} \text{ : } \{\mathcal{C}_{1},\mathcal{C}_{2},...,\mathcal{C}_{m} \}  $

   $ Z_{m} \gets 0 $

   \For{$i=1,2,...,m$}{
      \If{N-M\_Shattering($\mathcal{H}$,$\mathcal{C}_{i}$) == False}{
      $Z_{m} \gets Z_{m}+1 $
      }        
   }

   \If{$Z_{m} == m$}{
       \textbf{Return $d-1$}
   }   
}
\textbf{Return $+\infty$}
\end{algorithm}

After presenting all the theoretical foundations, we proceed to unveil our algorithm for computing the Vapnik-Chervonenkis (VC) dimension. 

For a given hypothesis class $H$, accuracy $\epsilon$, confidence $(1-\delta)$, and a large number $d_{\text{MAX}}$ (which reflects the case where the dimension of VC is infinite), we compute the sample size $m$ needed to assert that the shattering property evaluated in each iteration of our algorithm is Probably Approximately Correct (PAC) with probability $(1-\delta)$.

For each size $d$ in the range $[1, d_{\text{MAX}}]$, we generate $m$ random samples from the domain set $X$ with size $d$ (ie ${C_1, C_2, ..., C_m}$ where $C_i \subseteq X$ and $|C_i| = d$). Then, for each $C_i$, we apply the \textbf{Nechba-Mouhajir algorithm for shattering}\ref{algo_1:Nechba-Mouhajir algorithm for shattering} (which will be refreed as \textbf{N-M\_Shattering} ) to determine if $H$ shatters $C_i$. If $H$ does not shatter $C_i$, we increase $Z_m$ by 1. We repeat this process for all $C_i$ to compute the final $Z_m$, which probably and approximately, according to Hoeffding's inequality \ref{Hoeffding's inequality for our case study}, represents the probability of having the VC dimension less than $d$ (i.e., the probability of the event $\text{VCdim}(H) < d$). At this stage, we check if $Z_m = 1$, in which case we terminate the algorithm and output $\text{VC}(H) = d-1$. Otherwise, we perform the same process for the next iteration $d+1$, until we reach $d_{\text{MAX}}$, in which case we output $\text{VC}(H) = \infty$.

\begin{remark}
    It is complicated to calculate the error probability of Algorithm 2 \ref{alg:two}, but we can make this probability small by choosing very small values for $\epsilon$
and $\delta$.
\end{remark}

\subsection{The complexity of our algorithm}

To determine the complexity of our algorithm we should first determine the complexity of \textbf{N-M\_Sattering($\mathcal{H},\mathcal{C}$)} :

\begin{enumerate}
    \item \textbf{All\_combinations} is of size $2^{d}$
    \item The complexity of each iteration in the for-loop is the complexity of finding the output of the $ERM_H(S)$ and then computing its Emprirical loss $L_S$. We will denote it as follow :
     $$t = \mathcal{O}(L_{S}(ERM_{\mathcal{H}}(\mathcal{S})))$$ 
     
     \item in conclusion, the complexity of N-M\_Sattering($\mathcal{H},\mathcal{C}$) is : 
\[
\mathcal{O} ( 2^{d} \times t )
\]
    
\end{enumerate}

Now , we could determine the complexity of our \textbf{Mouhajir-Nechba algorithm for VC dimension} :
 \begin{enumerate}
     \item In worse case (i.e. if $\mathcal{H}$ has infinite dimension) we will loop over $d_{MAX}$ iterations.
     \item For each iteration, we will have a for-loop of complexity : \[
\mathcal{O} ( m\times 2^{d} \times t )
\]
     \item Thus, over all $d_{MAX}$ iterations the complexity will be :
      \[
\mathcal{O} ( \sum_{d=1}^{d_{MAX}} m\times 2^{d} \times t ) = \mathcal{O} ( m \sum_{d=1}^{d_{MAX}}  2^{d} \times t )
\]
     
 \end{enumerate}

\section{Experimentation} 
The objective of this section is to showcase the efficiency of our algorithm in computing the VC dimension and its applicability in real-world scenarios. To achieve this goal, we will utilize the \textbf{Mouhajir-Nechba algorithm for VC dimension} on pre-calculated mathematical classes. Specifically, we will focus on the Half-Spaces class as a case study.

\subsection{\textbf{Theoretical results}}

\par The VC dimension of half-spaces is formally established by Shalev-Shwartz and Ben-David in their book 'Understanding Machine Learning: From Theory to Algorithms'. Specifically, in Chapter 9, Section 9.1.3 titled 'The VC Dimension of Half-spaces' \cite{David14}, they prove the following theorem:

\begin{theorem}
The VC dimension of the class of nonhomogenous Half-Spaces in
$\mathbb{R}^d $ is $d + 1$.
\end{theorem}

\subsection{\textbf{Practical considerations}}

The major challenge in our experiment is to represent the ERM of Half-Spaces. Therefore, we will use Perceptron \cite{Rosenblatt58} as a specific implementation of this paradigm .

According to the work of Shalev-Shwartz and Ben-David in their book "Understanding Machine Learning: From Theory to Algorithms"\cite{David14}, the ERM rule for the class of half-Spaces can be represented in multiple ways, not just through the use of the Perceptron algorithm. In particular, sections 9.1.1 and 9.1.2 of their book describe alternative methods  :
  
  \begin{itemize}
      \item \textbf{For the realizable case}  solving the following linear program  : 
            \begin{align*}
                \max_x \quad u^T x \\
            Ax &\geq v \\
            \end{align*}
            where
            \begin{itemize}
                \item $u\in \mathbb{R}^d $  is a "dummy" vector : $u=(0,\ldots,0)$ (since every $w \in \mathbb{R}^d $ that satisfies the constraint is a solution to the ERM rule).
                \item $A$ is a matrix whose elements are $A_{i,j} = y_{i} x_{i,j}$, where $x_{i,j}$ is the $j^{\text{th}}$ component of the $i^{\text{th}}$ example within the sample $S=\{(x_1,y_1),\ldots,(x_m, y_m)$.
                \item $v=(1,\ldots,1) \in \mathbb{R}^m $ .
            \end{itemize}

        \item \textbf{For the Agnostic case}: solving the non-convex optimization problem:
\[h^{*}=\operatorname*{argmin}_{h}(L_{\mathcal{S}}(h)),\]
whereas, the major problem here is to handle the non-convexity of 0-1 loss: $l((x,h(x)),y)=\mathbb{1}_{\left(h(x)\neq y\right)}$. For this purpose, we can use the sub-gradient method.

  \end{itemize}

\subsection{Mouhajir-Nechba Algorithm for VC Dimension as a Tool for Evaluating ERM Implementation}

Using the \textbf{Mouhajir-Nechba algorithm for VC dimension}, we propose a methodology to assess the performance of a specific implementation of the ERM rule for a hypothesis class $\mathcal{H}$. 

Specifically, we determine the VC dimension of $\mathcal{H}$ under the given implementation and compare it with the theoretical results. If the results align within a certain tolerance level, we conclude that the implementation is good.

For instance, for the class of Half-Spaces, the Perceptron implementation is considered good if the Mouhajir-Nechba algorithm yields the same VC dimension as expected theoretically with a specified level of confidence.

\subsection{\textbf{Experimental results}}

\par The experiment that we have conducted, for ($\epsilon=10^{-2}$ and $\delta=10^{-2}$ ) has given the following results for both execution time and VC dimension :

\begin{table}[!h]
    \centering

\caption{VC-dimension of perceptron for dimension between 1 and 6 \label{tabularx}}
    \begin{tabularx}{1.12\textwidth} { 
  | >{\raggedright\arraybackslash}X 
  | >{\centering\arraybackslash}X
  | >{\centering\arraybackslash}X
  | >{\centering\arraybackslash}X
  | >{\centering\arraybackslash}X
  | >{\centering\arraybackslash}X
  | >{\raggedleft\arraybackslash}X | }
 \hline
 $dim(\mathcal{X})$ & 1 & 2 & 3 & 4 & 5 & 6  \\
 \hline
 VC-dimension  &2 & 3 & 4 & 5 & 6  &7  \\
\hline
execution time (seconds)  &34.76 & 46.78 & 70.12 & 98.64 & 103.53  &124.22\\
\hline
\end{tabularx}
 
\newpage

\end{table}
\begin{figure}[!h]
    \centering
    \includegraphics[scale=0.46]{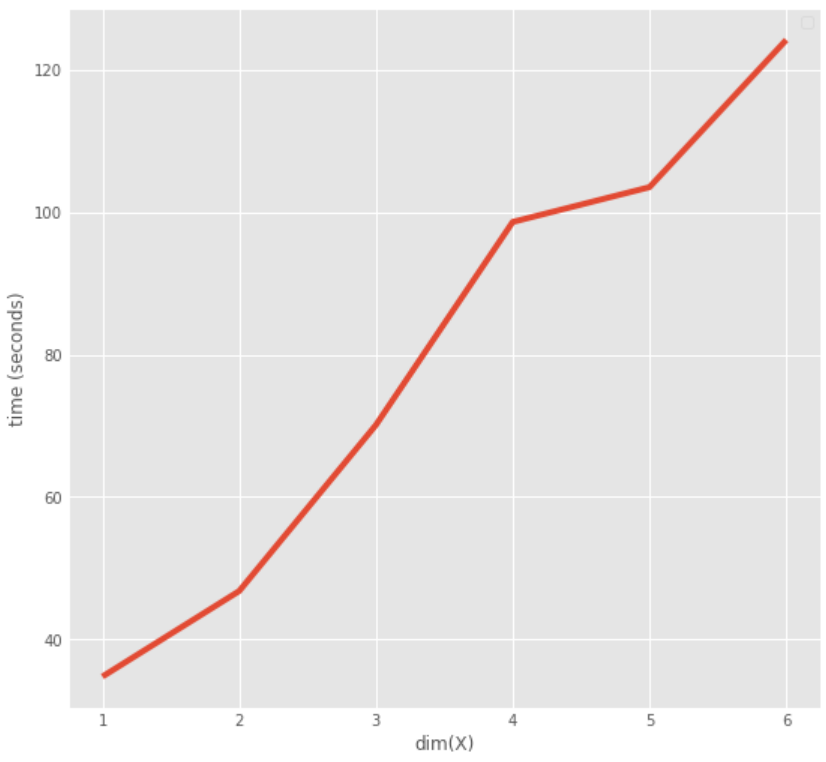}
    \caption{The computing time for Perceptron}
\end{figure}

The findings of our experiment can be summarized by two key observations. Firstly, the output of our algorithm for the VC dimension is in perfect alignment with the theoretical results. This indicates that our algorithm produces accurate and dependable outcomes, and that the Perceptron serves as a reliable implementation of the ERM principle for the half-spaces class. Secondly, we observed that the computational time required by our algorithm increases as the dimension of the domain set $\mathcal{X}$ increases. This highlights the computational complexity of our algorithm, which can pose a significant practical challenge. Nevertheless, we are confident that this issue can be overcome by developing a PRAM version of our algorithm that harnesses the processing power of acceleration devices.

\section{Summary and future work }

This paper addresses a fundamental challenge in the field of Machine Learning: computing the VC dimension for binary classification concept classes beyond the traditional discrete settings. To achieve this goal, we proposed a novel characterization of the shattering property based on the Empirical Risk Minimization paradigm. Our algorithm, which is built upon this characterization, can handle the problem of VC dimension without any constraints on the concept class or domain set.

However, the computational time required for our algorithm is an area of concern, which can be mitigated to some extent by harnessing the computing power of accelerators such as GPUs. Additionally, there may be challenges with our algorithm if the probability of shattered sets is extremely low or if there are issues with finite-precision arithmetic approximating concepts defined in terms of irrational numbers.

\bibliographystyle{plain}
\bibliography{sample-base}

\end{document}